\documentclass[journal]{IEEEtran}
\usepackage{cite}

\usepackage[pdftex]{graphicx}
\graphicspath{{./pdf/}}
\DeclareGraphicsExtensions{.pdf,.jpeg,.png}

\usepackage{amsmath}
\usepackage{amsfonts}
\newcommand{\bb}{\mathbf}
\usepackage{amsthm}
\theoremstyle{definition}
\newtheorem{definition}{Definition}

\newtheorem{proposition}{Proposition}
\newtheorem{example}{Example}

\usepackage{algorithm}
\usepackage{algorithmic}
\usepackage[dvipsnames]{xcolor}

\usepackage{hyperref}

\begin{document}
\title{An Image Source Method Framework \\for Arbitrary Reflecting Boundaries}

\author{Pierre~Quinton,~%\IEEEmembership{Student~Member,~IEEE,}
        Pablo~Mart\'inez-Nuevo,~%\IEEEmembership{Member,~IEEE,}
        and~Martin~B.~M{\o}ller~%\IEEEmembership{Student~Member,~IEEE~}% <-this % stops a space
\thanks{Pierre Quinton is with the School of Computer and Communication Sciences, \'Ecole Polytechnique F\'ed\'erale de Lausanne, 1015 Lausanne, Switzerland (e-mail: pierre.quinton@epfl.ch).

P. Mart\'inez-Nuevo and M. B. M{\o}ller are with the research department at Bang \&~Olufsen, 7600 Struer, Denmark (e-mail: pmn@bang-olufsen.dk; mim@bang-olufsen.dk)}% <-this % stops a space
\thanks{}% <-this % stops a space
\thanks{}}

\markboth{}%Transactions on Audio, Speech, and Language Processing}%
{Shell \MakeLowercase{\textit{et al.}}: Bare Demo of IEEEtran.cls for IEEE Journals}

\maketitle

% in the abstract or keywords. 150-250 words
\begin{abstract}
We propose a theoretical framework for the image source method that generalizes to arbitrary reflecting boundaries, e.g. boundaries that are curved or even with certain openings. Furthermore, it can seamlessly incorporate boundary absorption, source directivity, and nonspecular reflections. This framework is based on the notion of \textit{reflection paths} that allows the introduction of the concepts of \textit{validity} and \textit{visibility} of virtual sources. These definitions facilitate the determination, for a given source and receiver location, of the distribution of virtual sources that explain the boundary effects of a wide range of reflecting surfaces. The structure of the set of virtual sources is then more general than just punctual virtual sources. Due to this more diverse configuration of image sources, we represent the room impulse response as an integral involving the temporal excitation signal against a measure determined by the source and receiver locations, and the original boundary. The latter smoothly enables, in an analytically tractable manner, the incorporation of more general boundary shapes as well as directivity of sources and boundary absorption while, at the same time, maintaining the conceptual benefits of the image source method.
\end{abstract}


\IEEEpeerreviewmaketitle

\section{Introduction}
\IEEEPARstart{T}{he} behavior of waves in enclosures can be modeled by solving the wave equation subject to the appropriate boundary conditions. For very simple geometries and boundary conditions, it is possible to express the solution analytically in an explicit manner \cite{Morse:1986aa,Kuttruff:2016aa}. However, in more complex scenarios, this model becomes more cumbersome from an analytical and practical point of view. In room acoustics, for example, it is not very beneficial when the room shape is not of very simple geometry or the walls are nonrigid \cite{Jacobsen:2013aa}.

% The behavior of waves in enclosures can be modeled by solving the wave equation subject to the appropriate boundary conditions. For very simple geometries and boundary conditions, it is possible to express the solution analytically through an eigenfunction expansion \cite{morse_ingaard, kuttruff}. However, in more complex scenarios this model becomes more cumbersome from an analytical and practical point of view. In room acoustics, for example, it is not very beneficial when the room shape is not of very simple geometry or the walls are nonrigid \cite{1}. 
% Numerical methods seek to solve this problem by discretizing the environment and time or frequency – response using methods like finite difference time-domain, boundary element method, finite element method, or spectral element method \cite{pind_2019}. The downside of such methods is the computational complexity and memory requirement as the frequency increases and the associated discretization elements becomes comparably smaller.

Under physically meaningful assumptions, it is possible to use a simpler model where the concept of a sound path or sound ray is used instead of that of a wave \cite{Allen:1979aa,Kuttruff:2016aa}.\footnote{If the interaction between reflections is not of interest, as is the case when modeling e.g. reverberation time, \textit{energy-based} models like ray tracing can be applied \cite{Savioja:2015}. In the current work, the scope is the \textit{pressure-based} image source method.} Similar approaches are used in geometrical optics \cite{Pedrotti:2017}. Then, the behavior of sound rays in a closed room emitted by a given source and reflected off the corresponding surfaces can be described by the concept of image---or virtual---sources. This approach has been extensively used for solving partial differential equations \cite{Sommerfeld:1949ab}. In one of its versions, given a punctual and omnidirectional source in a room, the boundary effects are described by the associated set of virtual sources. 

% Under physically meaningful assumptions, it is possible to use a simpler model where the concept of a sound path or sound ray is used instead of that of a wave \cite{allen_1979, kuttruff}.\footnote{If only the statistical properties of the room is of interest, as is the case for modelling reverberation, models like ray tracing and particle tracing can be applied \cite{}. In the current work, the scope is the phase preserving image source method.} Similar approaches are used in geometrical optics \cite{pedrotti^3}. The behavior of sound rays in a closed room emitted by a given source and reflected off the corresponding surfaces can be described by the concept of image – or virtual – sources. This approach has been extensively used for solving partial differential equations \cite{summerfeld_1949}. In the special case of the image source method, given a punctual and omnidirectional source in a room, the boundary effects are described by the associated set of virtual sources.

The image source method provides a less abstract description that has enabled a more tractable theoretical analysis and more efficient simulations. For example, it is used in \cite{Ajdler:2006aa} in order to derive theoretical guarantees about sampling density and reconstruction of sound fields in simple enclosures. In \cite{Mignot:2013aa}, sparse recovery of the early part of the room impulse response is also based on the concept of virtual sources. Additionally, it has been utilized as a model, in an explicit \cite{Dokmanic:2011aa} and implicit manner \cite{Ribeiro:2011aa}, in order to infer certain room shapes.

In spite of its wide applicability both as a theoretical and computational model, it has not been fully extended to arbitrary reflecting boundaries. In \cite{Keller:1953aa}, it is studied the applicability of the image method for some polygons. The first instance of the image source method within the context of room acoustics was derived for rectangular rooms \cite{Allen:1979aa}. It was also shown therein that, given a punctual omnidirectional source in a rigid-wall rectangular room, the modal solution was the same as the one derived under the image source method. This model was later extended to polygonal rooms in \cite{Borish:1984aa} where the concepts of \textit{visibility} and \textit{validity} of virtual sources were introduced in order to accommodate the more complex derivation of the virtual sources. However, the algorithm proposed in \cite{Borish:1984aa} to obtain the set of virtual sources does not extend, for example, to rooms that are not closed---i.e. presenting wall-size openings to free space---, or those consisting of curved walls.

In this paper, we develop a framework based on the image-source model for the analysis of enclosures with arbitrary boundaries that also models absorption and directivity of sources. We introduce the concept of \textit{reflection paths} associated with virtual sources. By further redefining previous notions of \textit{visibility} and \textit{validity}, we are able to extend the derivation of virtual sources to arbitrary rooms---e.g. open or with curved walls---with possibly nonspecular reflections. The distribution of virtual sources is then no longer restricted to a discrete set of points and can consist of more complex structures, e.g. continuous contours in $\mathbb{R}^3$. In order to accommodate this potential heterogeneity, we also establish the foundations of a theoretical framework for expressing, in a closed and explicit manner, the room impulse response as an integral against an appropriately defined measure determined by the virtual sources distribution.

In Section \ref{sec:model}, we introduce, from first principles, sufficient definitions for a more general image source model. The main emphasis is to focus on \textit{reflection paths} in order to redefine \textit{visibility} and \textit{validity} of virtual sources. We still use the concept of virtual sources associated with a particular reflection path. We also show why previous methods cannot handle more general cases. A more familiar example of a room with planar walls is also presented to illustrate this framework. Section \ref{sec:RIRmeasure} describes the approach to obtain the room impulse response from a measure-theoretic point of view given the distribution of virtual sources.

We introduce here some notation that we will be using throughout the paper. Let $\mathcal{H}$ be a Hilbert space with inner product $\langle\cdot{,}\cdot\rangle$. We will often use $\mathcal{H}=\mathbb{R}^n$ with the usual inner product. The norm induced by this inner product is then given by $\lVert \bb{u} \rVert = \sqrt{\langle \bb{u} , \bb{u} \rangle}$ for all $\bb{u}\in\mathcal{H}$. The indicator function $\bb{1}_\mathcal{A}(\cdot)$ evaluates to $1$ if the argument is in the set $\mathcal{A}$ and $0$ otherwise. The nonnegative integers are denoted by $\mathbb{N}=\{0,1,2,\ldots\}$. For $i,j\in\mathbb{N}$ satisfying $i\leq j$, we denote the set of integers between $i$ and $j$ by $[i{:}j]=\lbrace k\in \mathbb{N} \mid i\leq k \leq j\textrm{ for }i,j\in\mathbb{N}\rbrace$.

% {\color{blue}
% Describe what this model could add to the image source method framework, for example absorption, refraction, non perfect walls, dense air.

% Possible uses: convenient for efficiently simulating room acoustics, analytically tractable since we can superimpose free-field responses, e.g. plenacoustic function, room shape inference \cite{Dokmanic:2011aa} \cite{Ajdler:2006aa}

% I think that the idea I want to convey is more about the generalization part. I would like that to be a framework to help proving that an algorithm is actually correct, or converge to the right solution, there will probably not be much emphases on practice, except a proof that Borish method is correct for convex polyhedral rooms.

% Kutruff considers the image source method as part of geometrical acoustics. I think since we cannot prove that the solution to the wave equation is the same with this virtual sources, we should go for a sound ray interpretation. Then, we can say that we know that for a rectangular room is the same as the modal solution.

% law of reflection formalized as valid reflection paths. Visibility for generalization.
% }

% {\color{red} Botish's focuses on validity and visibility of virtual sources; we focus on reflection paths and explaining a reflection path with a virtual source. Do we really need Dokmanics ref?}

\section{Framework for the Image Source Method}
\label{sec:model}
Within the context of wave propagation and under appropriate conditions, it is possible to explain reflection off boundaries of spherical waves in a geometrical way. This approach interprets wave fronts as rays originating at a certain point which we refer to as a \textit{source}. Reflection is usually modeled with the use of a vector associated with each point in the boundary. Under appropriate regularity conditions, this vector is chosen as the normal vector to the boundary. Loosely speaking, the \textit{law of reflection} then states that at the interface between two different media a ray reflects off the boundary in such a way that the angle of the incident ray equals the angle of the reflected ray when both angles are considered with respect to the normal to the boundary. Additionally, incident and reflected rays belong to the same plane \cite[Chapter 4]{Kuttruff:2016aa}. 

The main components of the law of reflection are thus the notion of a boundary and a normal vector. However, it is possible to generalize the definition of reflection by assigning an arbitrary vector to a point in the boundary. The incident and reflected angles are then considered with respect to this vector. Note that the latter does not necessarily have to be normal to the boundary. This definition allows us to circumvent restrictive regularity conditions at the boundary that are required for the existence of normal vectors.

Let us now formalize the concepts stated above. In particular, we define a boundary in a Hilbert space $\mathcal{H}$ as a set of points $\mathcal{B}\subseteq\mathcal{H}$. Then, reflections are explained with respect to a vector field determined by $\mathcal{B}$ in the following manner
\begin{equation}
\label{eq:BoundaryField}
\begin{split}
\bb{n}_\mathcal{B}:\ &\mathcal{H}\rightarrow \mathcal{H}\\
& \mathbf{u}\mapsto\mathbf{n}_\mathcal{B}(\bb{u})
\end{split}
\end{equation}
where $||\bb{n}_\mathcal{B}(\bb{u})||=\bb{1}_\mathcal{B}(\bb{u})$ for any $\bb{u}\in\mathcal{H}$. In other words, for some $\bb{v}\in\mathcal{H}$, $||\bb{n}_{\mathcal{B}}(\bb{v})||\neq0$ if and only if $\bb{v}$ is a point belonging to the boundary $\mathcal{B}$. For the sake of familiarity with the Euclidean space---i.e. consider for now that $\mathcal{H}=\mathbb{R}^n$---, it is possible to think of $\bb{n}_\mathcal{B}(\bb{u})$ as a normal vector to a hyperplane $\mathcal{S}=\{\bb{v}\in\mathbb{R}^n:\bb{n}_\mathcal{B}(\bb{u})^T(\bb{u}-\bb{v})=0\}$ where $\bb{u}\in\mathcal{S}$. This hyperplane then describes a virtual reflecting plane where the specular reflection at $\bb{u}$ takes place. Such an approach allows us to define reflections even for isolated single points (see Fig.~\ref{fig:VectorField}).

\begin{figure}[!ht]
\centering
\includegraphics[width=0.7\linewidth]{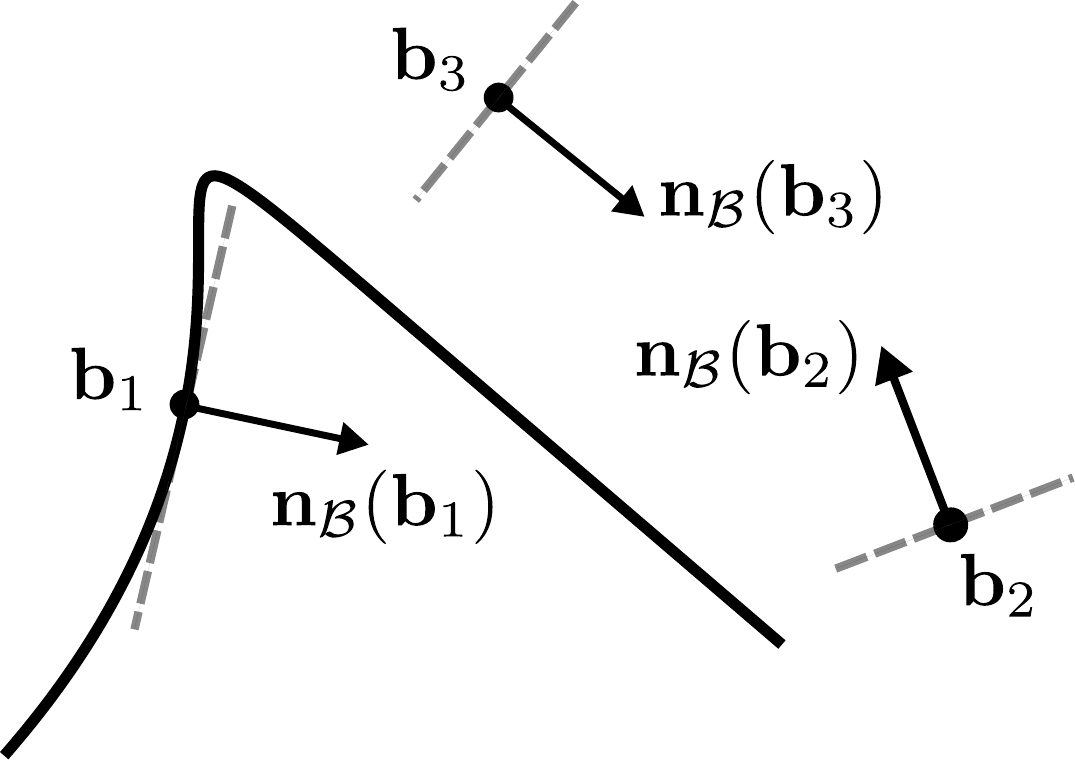}
\caption{Example of the vector field $\bb{n}_\mathcal{B}$ for $\mathcal{H}=\mathbb{R}^2$ where $\mathcal{B}$ consists of the solid curve and the points $\bb{b}_2$ and $\bb{b}_3$. The dashed lines represent the different lines corresponding to the interpretation of $\bb{n}_\mathcal{B}(\bb{b}_2)$ and $\bb{n}_\mathcal{B}(\bb{b}_3)$ as vectors orthonormal to the corresponding hyperplanes. Assuming the appropriate regularity conditions, $\bb{n}_\mathcal{B}(\bb{b}_1)$ can be chosen as a unit vector perpendicular to the tangent line of the curve at $\bb{b}_1$.}
\label{fig:VectorField}
\end{figure}

In our scenario, omnidirectional rays then originate at the source and are reflected from the boundary in a specular way according to $\bb{n}_{\mathcal{B}}$. Considering initially a single incident ray, the reflected ray can be explained by placing another source behind the boundary which is referred to as \textit{virtual} source or \textit{image} source \cite{Allen:1979aa,Borish:1984aa}. In particular, the location of the virtual source is formally given in the following definition by means of the notion of symmetric projection.
\begin{definition}
\label{def:VirtualSourceLocation}
Consider a boundary $\mathcal{B}\subseteq\mathcal{H}$. The symmetric projection of a vector $\bb{u}\in\mathcal{H}$ with respect to $\bb{v}\in\mathcal{B}$ is defined as follows
\begin{align}
\label{eq:SymmProjection}
P_{\bb{v}}(\bb{u}) := \bb{u} - 2 \langle \bb{u}-\bb{v} , \bb{n}_{\mathcal B}(\bb{v}) \rangle \bb{n}_{\mathcal B}(\bb{v})
\end{align}
\end{definition}
In other words, the position of the virtual source that explains the reflection of a ray at the point $\bb{v}$ is given by the symmetric projection $P_{\bb{v}}(\bb{u})$. Fig.~\ref{fig:Projection} illustrates this effect where, for illustration purposes, we only consider a single ray of an omnidirectional source. Note that it follows from Definition \ref{def:VirtualSourceLocation} that the distance between the source $\bb{u}$ and the reflection point $\bb{v}$ is the same as from the virtual source $P_{\bb{v}}(\bb{u})$ to $\bb{v}$, i.e. $||\bb{u}-\bb{v}||=||P_\bb{v}(\bb{u})-\bb{v}||$. Notice also that the symmetric projection is independent of the sign of the vector $\bb{n}_{\mathcal B}(\bb{v})$, i.e. we do not concern ourselves with any convention regarding inward- or outward-pointing vectors.

By generating these virtual sources, it is then possible to entirely replace the boundary by a set of virtual sources modeling the reflections. Depending on the geometry of the boundary, the distribution of sources takes different forms. For example, in a rectangular enclosure many of the reflections are explained by virtual sources that coincide in a single point \cite{Allen:1979aa}. However, as can be inferred from Fig.~\ref{fig:Projection}, the location of virtual sources for curved boundaries may lie in a continuous path. We formalize these notions in Appendix \ref{app:GeomResults}.

\begin{figure}[!ht]
\centering
\includegraphics[width=0.7\linewidth]{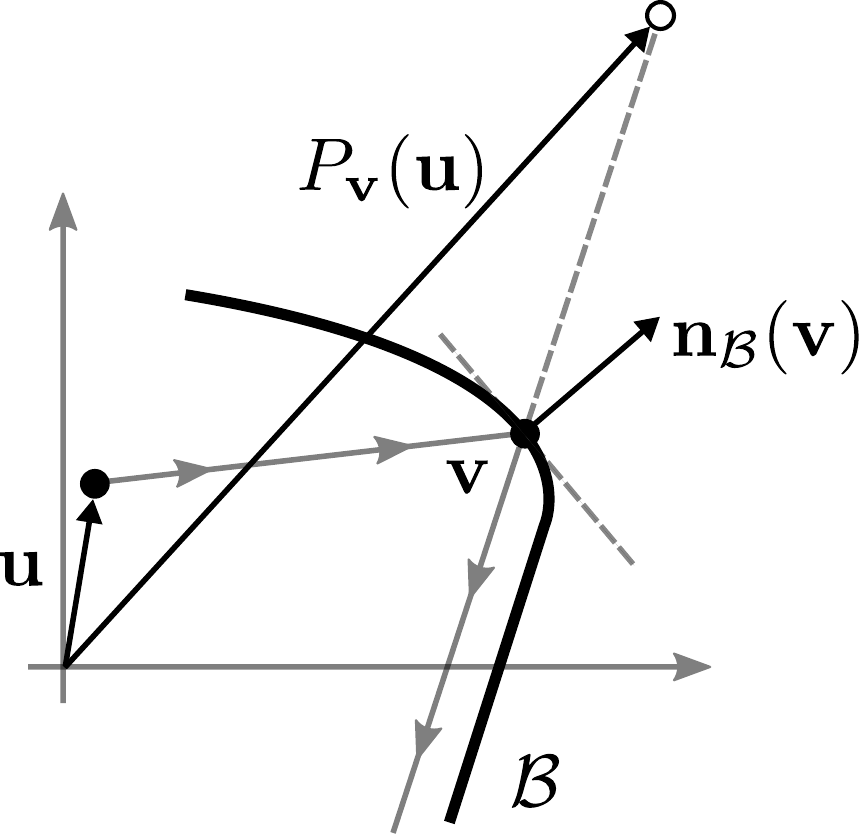}
\caption{An omnidirectional source is placed at $\bb{u}$ together with a solid curve $\mathcal{B}$ representing the boundary. The ray reflected at point $\bb{v}\in\mathcal{B}$ can be explained by a virtual source located at $P_{\bb{v}}(\bb{u})$ as depicted in the figure. The reflection is given with respect to the vector field $\bb{n}_{\mathcal{B}}(\cdot)$ which in this case assigns an outward-pointing unit vector normal to the curve at $\bb{v}$.}
\label{fig:Projection}
\end{figure}

It will be useful for our analysis later to introduce how we refer to a ray that has been reflected off different boundary points as we describe in Definition \ref{def:ReflectionPath}. For this definition, we temporarily drop any consideration about the law of reflection and simply consider a reflection path as a set of ordered points consisting of the locations where the ray is generated, where it is reflected, and where it is observed. We refer to this observation point as the \textit{sink}. 

\begin{definition}
\label{def:ReflectionPath}
Consider a boundary $\mathcal{B}\subseteq\mathcal{H}$. For $i\geq 0$ and distinct $\bb{y}_0, \dots, \bb{y}_{i+1} \in \mathcal{H}$, $(\bb{y}_0, \dots , \bb{y}_{i+1})$ is called a reflection path if $\bb{y}_j\in\mathcal{B}$ for any $j\in [1{:}i]$. The vectors $\bb{y}_1$, \dots , $\bb{y}_i$ are referred to as the reflection points and $\bb{y}_0, \bb{y}_{i+1}$ correspond to the source and sink respectively.
\end{definition}

Fig.~\ref{fig:Validity} shows two examples of reflection paths. It is important to emphasize that Definition \ref{def:ReflectionPath} includes reflection paths that do not conform with the law of reflection that we informally stated above. Thus, we introduce in the next definition the notion of a \textit{valid} reflection path. Essentially, valid reflection paths is our approach to formalizing the law of reflection.

\begin{definition}
\label{def:validity}
A reflection path $({\bb y}_0, \dots , {\bb y}_{i+1})$ is said to be valid if for any $j\in[1{:}i]$
\begin{align}
\label{eq:LawReflection}
\frac{{\bb y}_{j+1}-{\bb y}_{j}}{\lVert {\bb y}_{j+1}-{\bb y}_{j}\rVert} = \frac{{\bb y}_{j}-P_{\bb{y}_j}(\bb{y}_{j-1})}{\lVert {\bb y}_{j}-P_{\bb{y}_j}(\bb{y}_{j-1})\rVert}.
\end{align}
\end{definition}
Thus, Definition \ref{def:validity} guarantees that valid reflection paths consists of incident and reflected rays that form the same angle with respect to the hypersurface associated with the vector field at the reflection point. In order to see that, note that a valid reflection path $({\bb y}_0, \dots , {\bb y}_{i+1})$ is constructed such that the incident ray at $\bb{y}_j$ is in the the direction of $\bb{y}_j-\bb{y}_{j-1}$ and the reflected ray in $\bb{y}_{j+1}-\bb{y}_{j}$. These vectors both form the same angle $\theta_j$ with respect to $\bb{n}_{\mathcal{B}}(\bb{y}_j)$. In other words, we have that
\begin{equation}
\begin{split}
\cos\theta_j&=\Big\langle\frac{\bb{y}_j-\bb{y}_{j-1}}{||\bb{y}_j-\bb{y}_{j-1}||},\bb{n}_\mathcal{B}(\bb{y}_j)\Big\rangle\\
&=\Big\langle\frac{\bb{y}_{j+1}-\bb{y}_{j}}{||\bb{y}_{j+1}-\bb{y}_{j}||},-\bb{n}_{\mathcal{B}}(\bb{y}_j)\Big\rangle
\end{split}
\end{equation}
which is shown in Proposition \ref{prop:AngleValidity} in Appendix \ref{app:GeomResults}. Moreover, it remains to show that $\bb y_{j+1}-\bb y_j$, $\bb y_j - \bb y_{j-1}$, and $\bb n_{\mathcal B}(\bb y_j)$ belong to the same plane. By using Definition \ref{def:VirtualSourceLocation} and \ref{def:validity}, it is straightforward to see that $\bb y_{j+1}-\bb y_j$ results from a linear combination of $\bb y_j - \bb y_{j-1}$ and $\bb n_{\mathcal B}(\bb y_j)$, i.e.
\begin{equation}
\alpha({\bb y}_{j+1}-{\bb y}_{j})=({\bb y}_{j}- \bb{y}_{j-1})+\beta \bb{n}_{\mathcal B}(\bb{y}_j)
\end{equation}
where $\alpha=\lVert {\bb y}_{j}-P_{\bb{y}_j}(\bb{y}_{j-1})\rVert/\lVert {\bb y}_{j+1}-{\bb y}_{j}\rVert$ and $\beta=2 \langle \bb{y}_{j-1}-\bb{y}_j , \bb{n}_{\mathcal B}(\bb{y}_j) \rangle $. Thus, we formalize the law of reflection by introducing the equivalent notion of valid reflection paths.

Fig.~\ref{fig:Validity} shows a valid reflection path when the vector field is chosen to be orthonormal to the boundaries. Notice that it is always possible to choose $\bb{n}_{\mathcal{B}}$ differently so that $(\bb{y}_0,\bb{x}_1,\bb{x}_2,\bb{y}_3)$ is a valid reflection path.

\begin{figure}[!ht]
\centering
\includegraphics[width=0.9\linewidth]{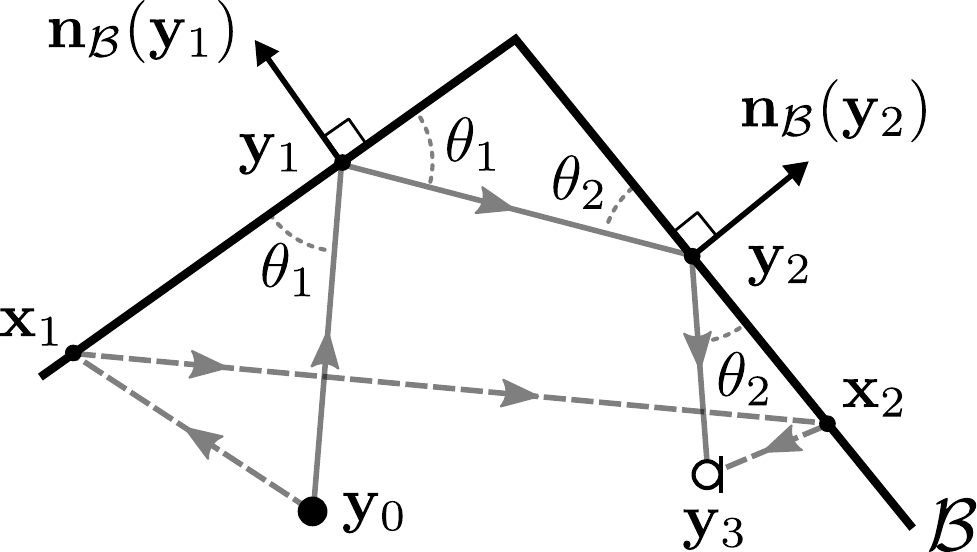}
\caption{Illustration of two reflection paths, i.e. $(\bb{y}_0,\bb{y}_1,\bb{y}_2,\bb{y}_3)$ and $(\bb{y}_0,\bb{x}_1,\bb{x}_2,\bb{y}_3)$ where $\bb{y}_0$ is the source and $\bb{y}_3$ the sink. Assume that at the reflection points, $\bb{n}_{\mathcal{B}}$ assigns an outward-pointing unit vector perpendicular to the corresponding boundaries. Then, the reflection path $(\bb{y}_0,\bb{y}_1,\bb{y}_2,\bb{y}_3)$ is valid according to Definition \ref{def:validity} whereas the path $(\bb{y}_0,\bb{x}_1,\bb{x}_2,\bb{y}_3)$ is clearly not valid.} 
\label{fig:Validity}
\end{figure}

\subsection{Visible Reflection Paths and Virtual Sources}
The main approach that we exploit in this paper is to focus on the properties of reflection paths instead of virtual sources. This interpretation is based upon the observation that each reflection path is explained by a single virtual source. Under the assumptions of the wave propagation model we are considering, the response at a sink $\mathbf{r}$ for a reflection path $\mathbf{R}=(\mathbf{s},\mathbf{y}_1,\ldots,\mathbf{y}_i,\mathbf{r})$ only depends on the distance traveled---further assuming perfectly reflecting boundaries; we will include absorption in Section \ref{sec:RIRmeasure}. In particular, this distance is simply given by 
\begin{equation}
\label{eq:DistanceReflection}
|\mathbf{R}|:=\sum_{j=0}^i||\mathbf{y}_{j+1}-\mathbf{y}_j||
\end{equation}
where $\mathbf{y}_0=\mathbf{s}$ and $\mathbf{y}_{i+1}=\mathbf{r}$. In principle, it is then possible to explain this reflection path by placing a virtual source at any point at a distance $|\mathbf{R}|$ from the sink, i.e. any point $\mathbf{v}\in\mathcal{H}$ such that $||\mathbf{v}-\mathbf{r}||=|\mathbf{R}|$.

One way of finding a point satisfying the latter is by means of recursively performing the symmetric projections introduced in Definition \ref{def:VirtualSourceLocation}. In particular, given a reflection path $\mathbf{R}=(\mathbf{s},\mathbf{y}_1,\ldots,\mathbf{y}_i,\mathbf{r})$, we can always find the associated virtual source as $\mathbf{u}=P_{\bb y_i}\circ \cdots \circ P_{\bb y_1}(\bb s)$. It follows from Definition \ref{def:VirtualSourceLocation} that $||\mathbf{u}-\mathbf{r}||=|\mathbf{R}|$. Thus, we can think of $\mathbf{u}$ as the virtual source that explains the reflection path $\mathbf{R}$. Moreover, it is the only position where the sound path from the virtual source has the same angle of incidence, with respect to the receiver, as the original reflection path. This will become relevant in Section \ref{sec:RIRmeasure} when we introduce directivity.

% \begin{equation}
% \begin{split}
% \nu_\mathcal{B}:\ &\mathcal{A}_\mathcal{B}\rightarrow \mathcal{H}\\
% & \mathbf{R}\mapsto\nu_{\mathcal{B}}(\mathbf{R})=P_{\bb y_i}\circ \cdots \circ P_{\bb y_1}(\bb s)
% \end{split}
% \end{equation}
% where $\mathcal{A}_\mathcal{B}$ is the set of valid reflection paths corresponding to the boundary $\mathcal{B}\subseteq\mathcal{H}$.

In other words, we associate a virtual source for each reflection path in such a way that the effect of the different boundary reflections for that particular path is completely explained by this virtual source. In order for the reflection paths to be physically meaningful, we introduced the concept of valid reflection paths. However, we can still have valid reflection paths that are not, in principle, physically realizable. Fig.~\ref{fig:Visibility} shows an example where the reflection path $(\mathbf{s},\mathbf{y}_4,\mathbf{y}_3^\prime,\mathbf{r})$, though valid, results in a ray that crosses the boundary. In order to avoid this, we now introduce the concept of visibility. This definition guarantees that the open line segment joining two subsequent reflection points does not intersect the boundary.
\begin{definition}
Consider a reflection path $({\bb y}_0, \dots , {\bb y}_{i+1})$ and $\bb v_j(\lambda):= \lambda {\bb y}_{j+1} + (1-\lambda){\bb y}_j$ for $j\in [0{:}i]$ and $\lambda\in(0,1)$. The reflection path $({\bb y}_0, \dots , {\bb y}_{i+1})$ is said to be visible if the following is satisfied
\begin{equation}
\label{eq:VisibilityCondition}
\langle {\bb n}_{\mathcal B}(\bb v_j(\lambda)), {\bb y}_{j+1}-{\bb y}_j\rangle = 0.
\end{equation}
for all $j\in [0{:}i]$ and all $\lambda\in(0,1)$.
\end{definition}
It follows from this definition that visibility is independent of validity, e.g. a reflection path can be visible and not valid. Fig.~\ref{fig:Visibility} depicts two valid reflection paths where $(\mathbf{s},\mathbf{y}_1,\mathbf{y}_2,\mathbf{y}_3,\mathbf{y}_4,\mathbf{r})$ is visible and $(\mathbf{s},\mathbf{y}_4,\mathbf{y}_3^\prime,\mathbf{r})$ is not visible. It is straightforward to see that there exists a $\lambda\in(0,1)$ such that $\langle {\bb n}_{\mathcal B}(\lambda {\bb r} + (1-\lambda){\bb y}_3^\prime), {\bb r}-{\bb y}_3^\prime\rangle \neq 0$.

\begin{figure}[!ht]
\centering
\includegraphics[width=0.8\linewidth]{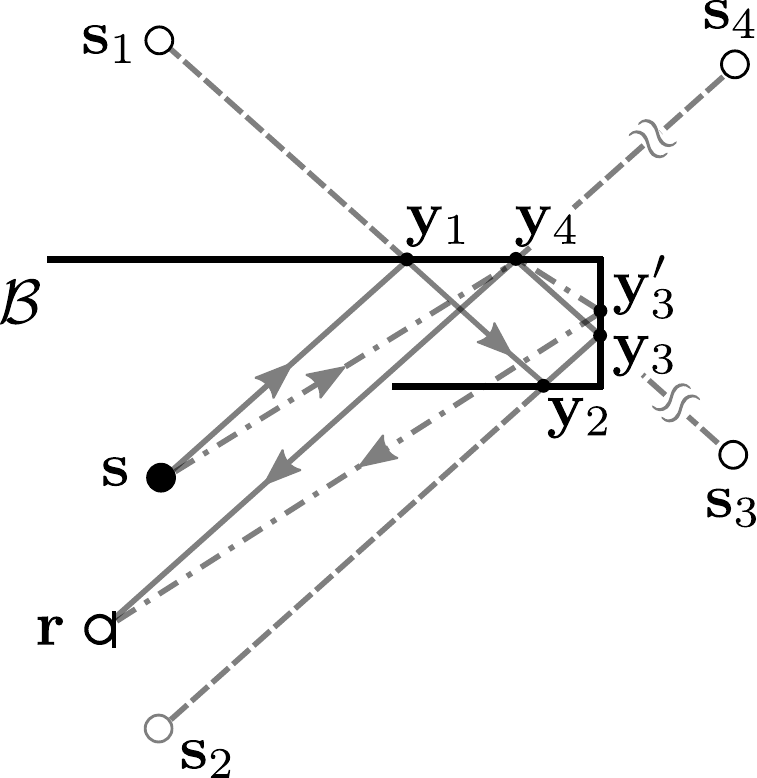}
\caption{Example of a valid and visible reflection path $(\mathbf{s},\mathbf{y}_1,\mathbf{y}_2,\mathbf{y}_3,\mathbf{y}_4,\mathbf{r})$, and a valid reflection path $(\mathbf{s},\mathbf{y}_4,\mathbf{y}_3^\prime,\mathbf{r})$ which is not visible where $\{\bb s_i\}$ are given by symmetric projections. The relative positions of $\mathbf{s}_3$ and $\mathbf{s}_4$ have been modified for illustration purposes.}
\label{fig:Visibility}
\end{figure}

Thus, in our approach, we consider valid and visible reflection paths---for a given room geometry---and generate the virtual sources based on these reflection paths. This is in contrast with previous literature \cite{Borish:1984aa} where virtual sources are generated based on the room geometry in order to explain reflection paths for a given source and sink. The definitions presented in this paper allow us to generalize the image source method to arbitrary reflecting boundaries which the approach in \cite{Borish:1984aa} does not cover, e.g. the boundary described in Example \ref{ex:NoSoundCorridor}.
\begin{example}[No-sound Corridor Effect]
If we consider, for example, a scenario where rays can reflect off both sides of a wall, the method presented in \cite{Borish:1984aa} is not able to appropriately accommodate this situation. In particular, the latter approach requires to establish a convention regarding inward-pointing normal vectors in order to explain reflections. Thus, virtual sources that have been generated by means of outward-pointing normal vectors are discarded from the model. This can lead to neglecting physically meaningful reflection paths. For example, in Fig.~\ref{fig:NoSoundCorridor}, assuming $\hat{\mathbf{n}}$ is an inward-pointing vector, $\hat{\mathbf{n}}^\prime$ is then pointing outwards. This means that the reflection path $(\mathbf{s},\mathbf{y}_2,\mathbf{r})$ is discarded from the model. By changing the inward- and outward-pointing convention of these two vectors, we find a degenerate situation where no rays are present in a region between the two walls---i.e. rays such as $(\mathbf{s},\mathbf{y}_1,\mathbf{y}_2,\mathbf{y}_3,\mathbf{r})$ are neglected. If $\mathbf{s}$ is a sound source, we refer to this as the no-sound corridor effect.

\label{ex:NoSoundCorridor}
\end{example}
\begin{figure}[!ht]{}
\centering
\includegraphics[width=0.9\linewidth]{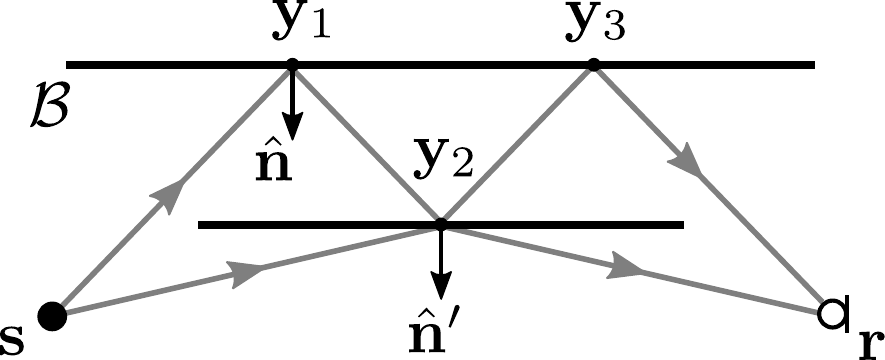}
\caption{No-sound corridor scenario described in Example \ref{ex:NoSoundCorridor}. Any convention of normal vectors---i.e. inward- or outward-pointing vectors---leads to a situation where either one of the reflection paths shown is discarded according to previous approaches \cite{Borish:1984aa}.}
\label{fig:NoSoundCorridor}
\end{figure}

\subsection{Rooms with Planar Walls}{}
In practice, it is very common to use the image source method in polygonal enclosures or rooms, i.e. with planar walls. In principle, the image source method presented here applies to any boundary as defined in (\ref{eq:BoundaryField}). However, rooms consisting of planar walls present particular characteristics that make them convenient for analysis. We do not restrict ourselves to closed rooms, but we consider a more general type of room which, for example, are also allowed to be open and with finite- or infinite-length walls (see Fig.~\ref{fig:Visibility} as an example of a room as referred herein). We formalize these notions in the following definition.

%unit vectors $\{\mathbf{n}_i\}_{i\in\mathcal{I}}$, and constants $\{b_i\}_{i\in\mathcal{I}}$ for some $\mathcal{I}\subseteq\mathbb{N}$ such that $V_i\subseteq S_i$ and

\begin{definition}
\label{def:planarB}
A boundary $\mathcal B\subset\mathbb{R}^n$ is planar if there exist a countable collection of connected and disjoint open sets $\{W_i\}_{i\in\mathcal{I}}$ and hypersurfaces $\{H_i\}_{i\in\mathcal{I}}$ such that 
\begin{equation}
\label{eq:planar}
\mathcal{B}\setminus B= \bigcup_{i\in\mathcal I} W_i
\end{equation}
where $W_i\subseteq H_i$, $H_i=\{\mathbf{v}\in\mathbb{R}^n|\langle \bb n_i, \bb v \rangle = b_i,\ \mathbf{n}_i\in\mathbb{R}^n\textrm{ and }b_i\in\mathbb{R}\}$ for $i\in\mathcal{I}$, and $B$ is a set of measure zero with respect to the Lebesgue measure. Each set $W_i$ is referred to as a wall whenever $|\mathcal{I}|\leq|\mathcal{I}^\prime|$ for any other collection of sets $\{W_i^\prime\}_{i\in\mathcal{I}^\prime}$ satisfying (\ref{eq:planar}).
\end{definition}
Roughly speaking, the walls of a planar boundary can be considered to be connected disjoint open subsets of hypersurfaces in the corresponding dimension. Note that single points are not considered walls. Further, Definition \ref{def:planarB} describes a more general set of rooms such as walls that can be circles in a three-dimensional space, which may be counterintuitive at first. Another example of a boundary in $\mathbb{R}^2$ is the union of the intervals $(0,1)$ and $(1,2)$, which are considered two different walls according to the definition above. Further, an appropriate choice of $B$ can make any partition of these intervals into open subintervals form a collection satisfying (\ref{eq:planar}); however, they cannot be referred to as walls. The introduction of the set $B$ also becomes relevant when we have, for example, intersecting walls. Consider a planar boundary consisting of two intersecting hypersurfaces. Then, $B$ is the intersection of the two hypersurfaces, and the walls correspond to four connected disjoint open sets satisfying that the intersection of their closure is precisely $B$. 

In the case of planar boundaries---e.g. a polyhedral room---it is possible to characterize validity by first considering the ordered sequence of walls where a ray reflects off for a given source $\mathbf{s}$ and sink $\mathbf{r}$. This stems from the fact that there exists a one-to-one correspondence between the associated virtual sources---which form a discrete set by Proposition \ref{prop:DiscreteVirtualPlanar}---and the sequences of walls corresponding to the trajectory of reflections.

In particular, assuming a collection of walls $\{W_i\}$ and a reflection path with a sequence of reflections at walls indexed by $(i_1,\dots,i_k)$, the single associated virtual source explaining these reflections can then be generated as $P_{\bb y_k}\circ \cdots \circ P_{\bb y_1}(\bb s)$ for any $\mathbf{y}_j\in \mathcal B\cap W_{i_j}$ where $j=1,\dots,k$. It is also possible to obtain the set of lines containing the reflection path. These are given by $\bigcup_{i=0}^{k} \mathcal L_i$ where 
\begin{equation}
\label{eq:lines1}
\begin{split}
\mathcal L_0=\lbrace \bb \lambda &\bb a + (1-\lambda) \bb b\ |\ \lambda\in\mathbb R,\\
&\bb a = P_{\bb y_{k}}\circ \cdots \circ P_{\bb y_1}(\bb s),\ \bb b = \bb r\}\\
\mathcal L_k=\lbrace \bb \lambda &\bb a + (1-\lambda) \bb b\ |\ \lambda\in\mathbb R,\\
&\bb a = \bb s,\ \bb b = P_{\bb y_{1}}\circ \cdots \circ P_{\bb y_k}(\bb r)\}
\end{split}
\end{equation}
and 
\begin{equation}
\label{eq:lines2}
\begin{split}
\mathcal L_i=\lbrace \bb \lambda &\bb a + (1-\lambda) \bb b\ |\ \lambda\in\mathbb R, \\
&\bb a = P_{\bb y_{k-i}}\circ \cdots \circ P_{\bb y_1}(\bb s), \\
&\bb b = P_{\bb y_{k-i+1}}\circ \cdots \circ P_{\bb y_k}(\bb r) \rbrace.
\end{split}
\end{equation}
for $i=1,\ldots,k-1$.

Then, a reflection path $(\bb s, \bb y_1,\dots,\bb y_k,\bb r)$ with reflection points at the corresponding sequence of walls is valid if and only if $W_i\cap \mathcal L_{i-1} \cap \mathcal L_{i}=\{\mathbf{y}_i\}$ for all $i\in[1{:}k]$ (see Proposition \ref{prop:ValidityLines} in Appendix \ref{app:GeomResults}). 

Algorithm \ref{alg:VSgeneration} summarizes how valid and visible virtual sources are generated for a planar boundary $\mathcal{B}$ consisting of walls $\{W_i\}_{i\in\mathcal{I}}$. In order to describe the procedure, it is convenient to define an auxiliary function parametrized by the number of consecutive reflections off distinct walls prior to arriving at the sink. This function, given a sequence of walls, either provides a valid reflection path if it is feasible for this combination or outputs a predefined value signaling that no valid reflection path is feasible. In particular, for a positive integer $k$, let us define the function $\bb \Psi^{(k)}:\mathcal I^k_{\mathcal{B}}\rightarrow \lbrace \varepsilon \rbrace \cup \mathcal H^{k+2}$ for some $\varepsilon\in\mathbb{R}$ which takes the form
\begin{equation}
\bb \Psi^{(k)}(\bb W)={}
\begin{cases} 
      (\bb s, \bb y_1,\dots,\bb y_k,\bb r), &\textrm{if }\ W_{i_j} \cap \mathcal L_{j-1}\cap \mathcal L_j = \{ \bb y_j\}\\
       \varepsilon, &\textrm{otherwise}
\end{cases}
\end{equation}
where $\bb W=(i_1,\dots,i_k)$. For a planar boundary with $N$ walls, the function $\bb \Psi^{(k)}$ simply checks if the reflections off subsets of $k$ walls are valid according to (\ref{eq:lines1}) and (\ref{eq:lines2}).

 \begin{algorithm}[H]
 \caption{Algorithm for generating virtual sources for a planar boundary.}
 \begin{algorithmic}[1]
 \label{alg:VSgeneration}
 \renewcommand{\algorithmicrequire}{\textbf{Input:}}
 \renewcommand{\algorithmicensure}{\textbf{Output:}}
 \REQUIRE Planar boundary $\mathcal B$, walls $\{W_i\}$, and $\bb s,\bb r$.
 \ENSURE Set of virtual sources $\mathcal{S}$
  \STATE $\mathcal{S}\leftarrow \lbrace \rbrace$
  \FOR {$k = 1,2,\dots$}
  \FOR {$\mathbf{W}\in\mathcal I^k$}
  \STATE $\mathbf{P}\leftarrow \mathbf{\Psi}^{(k)}(\mathbf{W})$
  \IF {($\mathbf{P} \ne \varepsilon$ \AND $\mathbf{P}$ visible)}
  \STATE $(\bb s,\bb y_1,\dots,\bb y_k,\bb r)\leftarrow \mathbf{P}$
  \STATE $\mathcal{S}\leftarrow \mathcal{S}\cup \lbrace P_{\bb y_k}\circ \cdots \circ P_{\bb y_1}(\bb s) \rbrace$
  \ENDIF
  \ENDFOR
  \ENDFOR
 \RETURN $\mathcal{S}$
 \end{algorithmic}
 \end{algorithm}

% {\color{red}discretization of the angles to check if the rays are visible or not?}

\section{Room Impulse Response for Arbitrary Room Geometries}
\label{sec:RIRmeasure}
The convenience of approaches like the image source method in acoustics \cite{Allen:1979aa} lies in its analytical and computational simplicity for computing the room impulse response (RIR) without having to explicitly solve differential equations. Given the room geometry, and the source and receiver locations, we have shown how to appropriately obtain the set of virtual sources so that it covers more general room geometries than previous methods \cite{Borish:1984aa}. For ease of explanation, we will consider Euclidean spaces, i.e. $\mathcal{H}=\mathbb{R}^N$ for $N$ a positive integer.

From our previous discussion, the appropriate virtual sources to consider are those that correspond to valid and visible reflection paths. In other words, we say that, given source and receiver locations, a virtual source $\bb u\in\mathcal{H}$ is valid and visible if there exists a valid and visible reflection path $(\mathbf{s},\mathbf{y}_1,\ldots,\mathbf{y}_k,\mathbf{r})$ such that $\bb u = P_{\bb y_{k}}\circ \cdots \circ P_{\bb y_1}(\bb s)$ given a boundary $\mathcal{B}$, source $\bb s$, and receiver $\bb r$. We denote the set of valid and visible virtual sources by $\mathcal{S}_\mathbf{r}$ where the dependence on $\bb s$ and $\mathcal{B}$ is implicit.

If the room consists of planar boundaries, the set of virtual sources $\mathcal{S}_\mathbf{r}\subset\mathbb{R}^N$ is composed of a discrete set of points, i.e. $\mathcal{S}_\mathbf{r}=\{\mathbf{s}_m\}_{m\in\mathbb{N}}$. In this case, the response to an excitation signal $f\in L^2(\mathbb{R})$ for a source $\mathbf{s}$ and receiver $\mathbf{r}$ can be written as \cite{Allen:1979aa}
\begin{equation}
\label{eq:DiscreteVS}
h_f(t;\mathbf{r})=\sum_{m\in\mathbb{N}}\frac{1}{||\mathbf{s}_m-\mathbf{r}||}f(t-||\mathbf{s}_m-\mathbf{r}||/c)
\end{equation}
where $c$ is the speed of the wavefronts, $\mathbf{s}_0:=\mathbf{s}$, and $\mathbf{s}\neq\mathbf{r}$.

Equation (\ref{eq:DiscreteVS}) is clear when the set of virtual sources is discrete, i.e. we have point sources. However, when considering arbitrary boundaries, the structure of the associated virtual sources may be more complex. Moreover, this model is not amenable to including boundary absorption or directivity of sources. In order to remedy this, it is necessary to use different analytical tools to compute the impulse response while keeping the advantages of an image source method.

The approach we will be taking here relies on expressing the impulse response as the integral of a function against an appropriate measure. This allows us to extend, in an analytically tractable manner, the results on point sources to more general distributions of virtual sources that describe boundaries with arbitrary shapes and absorption coefficients as well as incorporating the directivity of sources. 

For the sake of illustration, assume that the source and receiver are not collocated, i.e. there exists an $\epsilon>0$ such that $D_{\mathbb{R}^N}(\mathbf{r},\epsilon)\cap(\{\mathbf{s}\}\cup\mathcal{S})=\emptyset$ where $D_{\mathbb{R}^N}(\mathbf{r},\epsilon)$ is the open ball in $\mathbb{R}^N$, radius $\epsilon$, and center $\mathbf{r}$. It is then possible to write for $\mathcal{W}_{N,\epsilon}:=\mathbb{R}^N\setminus D_{\mathbb{R}^N}(\mathbf{s},\epsilon)$
\begin{equation}
\label{eq:IRmeasure}
h_f(t;\mathbf{r})=\int_{\mathcal{W}_{N,\epsilon}}\frac{1}{||\bb u-\bb r||}f(t-||\bb u-\bb r||/c)\mathrm{d}\mu_{\mathcal{S}_\mathbf{r}}(\bb u)
\end{equation}
where $\mu_{\mathcal{S}_\mathbf{r}}$ is an appropriately defined measure that depends on the set of virtual sources $\mathcal{S}_\mathbf{r}$. In a sense, this choice of impulse response assumes that the different virtual sources are in the far field with respect to the receiver location. The parameter $\epsilon$ can then be interpreted as modeling a distance from which this may be valid. Note that if source and receiver are collocated, it is straightforward to see that we can find an $\epsilon>0$ such that $D_{\mathbb{R}^N}(\mathbf{s},\epsilon)\cap(\mathcal{S}\setminus\mathbf{r})=\emptyset$ and then add the contribution of $f(t)$ directly. 

Under this interpretation, the example in (\ref{eq:DiscreteVS}) can be obtained by first defining the following measure
\begin{equation}
\mu_{\mathcal{S}_\mathbf{r}}(\Gamma):=\sum_{m\in\mathbb{N}}\delta_{\mathbf{s}_m}(\Gamma)
\end{equation}
where the terms in the summation are Dirac measures and $\Gamma$ is, for example, an element of the Borel $\sigma$-algebra over $\mathbb{R}^N$ denoted by $\mathfrak{B}_N$. It is clear then that using this measure in (\ref{eq:IRmeasure}) gives (\ref{eq:DiscreteVS}).

However, if the structure of the room also consists of nonplanar boundaries---e.g. curved walls---, the set of virtual sources is not discrete. In a three-dimensional setting, for example, we then wish to construct a model that allows the inclusion of virtual sources composed of points, curves, and surfaces in order to accommodate a wide variety of reflecting objects and sources distributions. One possibility to construct a measure that accommodates these cases is to first define it with the help of a limiting procedure. In particular, we consider punctual virtual sources that in the limit completely cover the set $\mathcal{S}_\mathbf{r}$. We illustrate this intuition in the next example.

\begin{example}
Assume that $\mathcal{S}_\bb r=(0,1)$, then we can define the following measure
\begin{equation}
\label{eq:DiracMeasuresUI}
\mu_{\mathcal{S},M}(\Gamma):=\frac{1}{M}\sum_{m=1}^{M-1}\delta_{\frac{m}{M}}(\Gamma\cap\mathcal{S})
\end{equation}
for $\Gamma$ in the Borel $\sigma$-algebra over $\mathbb{R}^N$ denoted by $\mathfrak{B}_N$. Thus, it can be shown that in the limiting case when $M\geq1$ is large, this combination of Dirac measures tend to represent the length of the interval $(0,1)$ (see Example \ref{ex:interval} and Proposition \ref{prop:DiracMeasures} in Appendix \ref{app:HMeasures}).
\begin{figure}[!ht]
\centering
\includegraphics[width=0.6\linewidth]{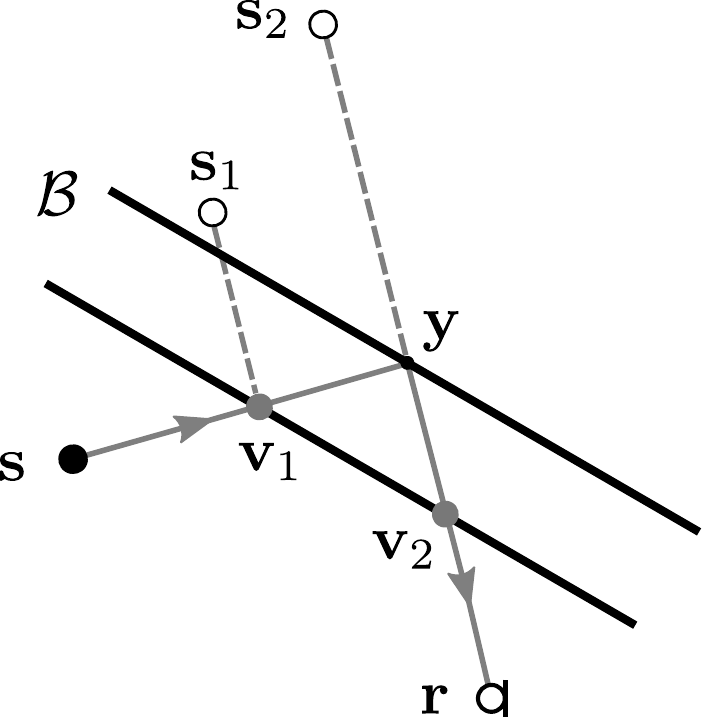}
\caption{Illustration of a degenerate situation for a boundary $\mathcal{B}$ consisting of two parallel lines where one of them is missing points $\bb v_1$ and $\bb v_2$. Thus, the reflection path $(\bb s, \bb y,\bb r)$ results in a valid and visible source $\bb s_2$ that is not physically relevant. In this scenario, it is only necessary to consider the virtual source $\bb s_1$ to explain the corresponding reflections.}
\label{fig:DegenerateEx}
\end{figure}
\end{example}

Proposition \ref{prop:DiracMeasures} motivates the use of a measure in (\ref{eq:IRmeasure}) that provides the length or area of the corresponding sets---note that there are different constructions of (\ref{eq:DiracMeasuresUI}) that also converge to the same measure. Thus, we require a measure that in an $N$-dimensional Euclidean space provides a sense of size of different sets as well as being able to distinguish among sets with different dimensions. The Hausdorff measure and the Hausdorff dimension exactly meet these requirements.

Some of the background on the Hausdorff measure can be found in Appendix \ref{app:HMeasures}. This measure provides a generalization of the Lebesgue measure---denoted by $\lambda_{\mathbb{R}^N}$ for the $N$-dimensional Euclidean space---in the sense that it does not depend on the dimension of the space that it operates on. This manifests itself in the Hausdorff measure detecting sets that are otherwise neglected by the Lebesgue measure. In other words, it does not obviate the lower dimensional sets, e.g. $\mathbf{H}^{N-1}(\mathbb{S}^{N-1})\neq0$ as opposed to $\lambda_{\mathbb{R}^N}(\mathbb{S}^{N-1})=0$ where $\mathbb{S}^{N-1}$ is the unit sphere in $\mathbb{R}^{N}$ \cite[Theorem 4.2.7]{Stroock:2011aa}. The way of formalizing this notion of dimension is based on the Hausdorff dimension which, in this case, gives $\mathrm{Hdim}(\mathbb{S}^{N-1})=N-1$.

The boundary $\mathcal{B}$, as defined so far, can lead to considering reflection paths that are not physically meaningful. For example, Fig.~\ref{fig:DegenerateEx} shows an example where the reflection path $(\bb s,\bb y,\bb r)$ results in a virtual source $\bb s_2$ that is not relevant in the final computation. In this case, it is only required to consider the virtual source $\bb s_1$ to correctly explain the reflections. In order to avoid these degenerate cases, we assume that the boundary is given by $\mathcal{B}=\bigcup_{n=0}^N\mathcal{C}^n$ where $\mathrm{Hdim}(\mathcal{C}^n)=n$ and there do not exist $\{\Gamma^n\in\mathfrak{B}_N\}$ such that $\mathbf{H}^n(\Gamma^n\setminus\mathcal{C}^n)=0$ whenever $\Gamma^n\supset\mathcal{C}^n$.

Then, the Hausdorff dimension provides us with a way of classifying the set of virtual sources $\mathcal{S}_\mathbf{r}\in\mathfrak{B}_N$ based on their dimension. In particular, we can write
\begin{equation}
\mathcal{S}_\mathbf{r}=\bigcup_{n=0}^N\mathcal{S}^n_\mathbf{r}
\end{equation}
where $\mathrm{Hdim}(\mathcal{S}^n_\mathbf{r})=n$. For example, punctual sources are then contained in $\mathcal{S}^0_\mathbf{r}$, curves in $\mathcal{S}^1_\mathbf{r}$, or surfaces in $\mathcal{S}^2_\mathbf{r}$. 

Building upon the concepts above, we then introduce the following measure that we can use to model the room impulse response in (\ref{eq:IRmeasure}) for arbitrary boundaries, i.e.
\begin{IEEEeqnarray}{rCl}
\label{eq:measure}
\mu_{\mathcal{S}_\bb r}(\Gamma)&:=&\sum_{n=0}^N\mathbf{H}^n(\Gamma\cap\mathcal{S}^n_\mathbf{r})\nonumber\\
&=&\sum_{m\in\mathcal{M}}\delta_{\mathbf{s_m}}(\Gamma)+\sum_{n=1}^N\mathbf{H}^n(\Gamma\cap\mathcal{S}^n_\mathbf{r})
\end{IEEEeqnarray}
for any $\Gamma\in\mathfrak{B}_N$ where $\mathcal{S}^0_\mathbf{r}=\{\mathbf{s}_m\}_{m\in\mathcal{M}}$ for $\mathcal{M}\subseteq\mathbb{N}$. Without loss of generality, we have assumed that we have a single punctual source. The extension to multiple punctual sources can be performed, in a straightforward manner, by superposition of the associated measures in the above equation.

% {\color{orange}
% \subsection{Diffraction}
% It has been shown that edge diffraction effects can be integrated into the image source model by adding a set of supplementary sources \cite{Torres:2001aa,Mechel:2002aa}. For example, in $\mathbb{R}^3$, the curve representing the intersection of two surfaces can cause diffraction phenomena. This can be modeled by continuously adding secondary sources distributed along the intersection. Under the framework presented here, the image sources for such a distribution can be obtained by considering each point in the curve as a source. It is straightforward to verify then that this results in image sources also distributed along curves, i.e. with the same Hausdorff dimension. In principle, this approach also extends to higher dimensions---under the physically correct interpretation---which allows us to include the virtual sources corresponding to diffraction effects in the set $\mathcal{S}_\bb r$. 
% }

\subsection{Directivity and Absorption}
In previous sections, we have assumed that the source is punctual and radiates energy in an omnidirectional manner. Directivity is a common model to describe how sources may emit energy more or less concentrated in different directions. In this section, we show how this property can be incorporated into our model. 

A convenient and popular way of expressing the directivity pattern of a given set of sources is by considering a function defined on a sphere around the source distribution. The values of this function are then related to how concentrated energy is in a particular direction. In acoustics, this function is mainly expressed using spherical harmonics \cite[Chapter 6]{Williams:1999aa}. In particular, for a function $d:\mathbb{S}^{N-1}\to\mathbb{R}$ such that $d\in L^2(\mathbb{S}^{N-1})$, the spherical harmonics $Y_k\in L^2(\mathbb{S}^{N-1})$ provide a unique representation $d=\sum_{k\geq0}Y^{(k)}$ where the convergence is in the $L^2$ norm \cite[Chapter IV]{Stein:2016aa}. Familiar examples include Fourier series representations on the unit circle whenever $N=2$, or spherical harmonics involving Legendre polynomials in $\mathbb{R}^3$.

In order to include directivity into our model, we use the functions $d_\bb s$ and $d_\bb r$ to describe the directivity of the punctual source $\bb s$ and the receiver $\bb r$, respectively. From the construction of the symmetric projection of a vector, it is straightforward to see that, given $\bb u\in\mathcal{S}_{\bb r}$, there exists a unique valid and visible reflection path $\mathbf{R}=(\bb s, \bb y_1, \ldots, \bb y_k, \bb r)$ such that $\bb u = P_{\bb y_{k-1}}\circ \cdots \circ P_{\bb y_1}(\bb s)$ and $\bb r = P_{\bb y_{k}}(\bb u)$. Thus, we define, for convenience, the bijective function 
\begin{IEEEeqnarray}{rCl}
\chi_{\bb r}:\ &\mathcal{S}_\bb r\ &\to\ \mathcal{A}_\mathcal{B}\nonumber\\
\ &\bb u\ &\mapsto\ \chi_\bb r(\bb u)=\mathbf{R}
\end{IEEEeqnarray}
where $\mathcal{A}_\mathcal{B}$ is the set of visible and valid reflection paths corresponding to the boundary $\mathcal{B}$, source $\bb s$, and receiver $\bb r$. Note that again $\bb u=\chi_\bb r^{-1}(\mathbf{R})=P_{\bb y_{k-1}}\circ \cdots \circ P_{\bb y_1}(\bb s)$. The function $\chi_{\bb r}$, although not explicitly denoted, is also determined by the boundary and the source location. Then, we can define the directivity coefficient of the virtual sources as 
\begin{IEEEeqnarray}{rCl}
\label{eq:directivity}
d_{\mathcal{S}_\bb r}:\ &\mathcal{S}_\bb r\ &\to\ \mathbb R\\
&\bb u\ &\mapsto\ d_{\mathcal S_\bb r}(\bb u)=d_\bb s\left( \frac{\bb y_1-\bb s}{\lVert \bb y_1-\bb s\rVert} \right)\cdot d_\bb r\left( \frac{\bb y_k-\bb r}{\lVert \bb y_k-\bb r\rVert} \right)\nonumber
\end{IEEEeqnarray}
where $\chi_\bb r(\bb u)=(\bb s, \bb y_1,\dots,\bb y_k,\bb r)$. If the reflection path is $(\bb s, \bb r)$, then (\ref{eq:directivity}) should be appropriately understood by considering $\bb r$ and $\bb s$ instead of the undefined $\bb y_1$ and $\bb y_k$, respectively. It can be observed that the image source framework we have constructed allows the inclusion of directivity of virtual sources by just considering the direction towards the first and last reflection point. This is in contrast to other approaches considering rotations of spherical harmonics representations \cite{Samarasinghe:2018aa,Abhayapala:2019aa}.

We model absorption solely as a function of the spatial dimension, i.e. $a:\mathcal{B}\to[0,1]$. It would also be straightforward to include amplification in the latter expression by just considering $\mathbb{R}$ as the image of $a$. Similarly, under the above conditions and given the reflection path $\chi_\bb r(\bb u)=\mathbf{R}$ for a particular virtual source $\bb u\in\mathcal{S}_\bb r$, we can incorporate absorption into the model by defining a function 
\begin{IEEEeqnarray}{rCl}
\label{eq:absorption}
a_{\mathcal{S}_\bb r}:\ &\mathcal{S}_\bb r\ &\to\ [0,1]\nonumber\\
&\bb u\ &\mapsto\ a_{\mathcal{S}_\bb r}(\bb u)=\prod_{i=1}^ka(\bb y_i).
\end{IEEEeqnarray}

Then, we can modify (\ref{eq:IRmeasure}) in such a way that models both directivity and absorption as follows
\begin{equation}
\label{eq:IRmeasure_ad}
h_f(t;\mathbf{r})=\int_{\mathcal{S}_\bb r}\frac{a_{\mathcal S_\bb r}(\bb u)d_{\mathcal S_\bb r}(\bb u)}{||\bb u-\bb r||}f\Big(t-\frac{||\bb u-\bb r||}{c}\Big)\mathrm{d}\mu_{\mathcal{S}_\mathbf{r}}(\bb u).
\end{equation}
Note that we have chosen the domain to be precisely $\mathcal{S}_\bb r$ in order to simplify the definition of the absorption and directivity coefficients. It is important to emphasize that this model could be easily extended to include absorption dependent on the angle of incidence by building upon the approach used for directivity. In particular, given a reflection path, the angle of incidence for each reflection point could be computed from the properties of the symmetric projection. Then, the corresponding factor could be incorporated into each of the terms in (\ref{eq:absorption}).

\section{Conclusions}
We presented a framework that allows the model obtained by the image source method to accommodate arbitrary reflecting boundaries. The latter includes, for example, curved walls and rooms presenting wall-size openings. In order to handle this more general and complex distribution of virtual sources, we also showed how the room impulse response can be explicitly obtained. The latter can also model boundary absorption and source directivity. The work shown here can establish the foundation for applying the image source method to more complicated boundary configurations whereby analytical and computational advantages can be gained. 

\appendices
\section{}
\label{app:GeomResults}
%\begin{lemma}
%\label{lemma:DiscreteVirtualPlanar}
%Assume an omidirectional source. If the boundary is a polyhedron (define polyhedron as a set of hypersurfaces inequalities), then the virtual sources generated as in Definition \ref{def:VirtualSourceLocation} lie in a discrete set of points.
%\end{lemma}
%\begin{proof}
%Actually, it’s quite easy I realized. For a wall defined by a hyper surface $\mathcal{S}=\{\mathbf{v}\in\mathcal{H}:\mathbf{n}^T\mathbf{v}=b\}$ for some $b\in\mathbb{R}$ and $\mathbf{b}\in\mathcal{H}$, the location of the source is given by some $\mathbf{u}$ and the virtual source for a reflection at a point $\mathbf{v}_1$ is $P_{\mathbf{v}_1}(\mathbf{u})=(\mathbf{u}-2\langle \mathbf{u},\mathbf{n}\rangle\mathbf{n})+2\langle\mathbf{v}_1,\mathbf{n}\rangle\mathbf{n}$ but both terms are constant since $\mathbf{n}=\mathbf{n}_\mathcal{S}(v)$ for any $\mathbf{v}\in\mathcal{S}$. Then, the generalization to a polyhedron is straightforward, could you put that formally into the lemma? Then, also say that it’s not the case for curved walls and you just show a counterexample, does that make sense?
%\end{proof}

\begin{proposition}
\label{prop:DiscreteVirtualPlanar}
Assume omidirectional sources. If the boundary $\mathcal{B}$ is planar for some connected and disjoint open sets $\{W_i\}_{i\in\mathcal{I}}$ and hypersurfaces $\{H_i\}_{i\in\mathcal{I}}$ where $\mathcal{I}\subseteq\mathbb{N}$, then the set of virtual sources generated as in Definition \ref{def:VirtualSourceLocation} consists of a discrete set of points.
\end{proposition}
\begin{proof}
It follows directly from the definition of hypersurface that, given $H_i$, there exists an $\bb{n}_i$ and $b_i\in\mathbb{R}$ such that $\bb n_{\mathcal B}(\bb v_i)=\bb n_i$ and $\langle \bb n_i, \bb v_i\rangle=b_i$ for any $\bb v_i \in H_i$. Given a source position $\bb{u}\in\mathcal{H}$, we can then write the following
\begin{equation}
\begin{split}
P_{\bb v_i}(\bb u)&=\bb u - 2\langle \bb u - \bb v_i, \bb n_{\mathcal B}(\bb v_i)\rangle \bb n_{\mathcal B}(\bb v_i)\\
&=\bb u-2\langle \bb u , \bb n_i\rangle \bb n_i+2 b_i \bb n_i
\end{split}
\end{equation}
for any $\bb{v}_i\in H_i$. Thus, for a given source and hypersurface, there is a single associated $P_{\bb v_i}(\bb u)$.
%For $i\in \mathcal I\subseteq \mathbb N$, let $W_i$ be the $i$-th wall, since it is a straight wall, for any $\bb v_i \in W_i$, $\bb n_{\mathcal B}(\bb v_i)=\bb n_i$ and $\langle \bb n_i, \bb v_i\rangle=b_i$ are constant. Hence, for $\bb v_i\in W_i$, and for any $\bb u \in \mathcal H$, $P_{\bb v_i}(\bb u)=\bb u - 2\langle \bb u - \bb v_i, \bb n_{\mathcal B}(\bb v_i)\rangle \bb n_{\mathcal B}(\bb v_i)=\bb u-2\langle \bb u , \bb n_i\rangle \bb n_i+2 b_i \bb n_i$ only depends on $\bb v_i$ through the index of the wall. The associated virtual source corresponding to a reflection path is then a function of the walls against which the reflections takes place.% Finally, the set of finite sequence of elements of $\mathcal I$ is a countable set, even if $\mathcal I=\mathbb N$.
\end{proof}

\begin{proposition}
\label{prop:AngleValidity}
Any valid reflection path $(\bb{y}_0, \dots , \bb{y}_{i+1})$ with an associated vector field $\bb{n}_\mathcal{B}$ satisfies the following
\begin{equation}
\label{eq:EqAngles}
\Big\langle\frac{\bb{y}_{j}-\bb{y}_{j-1}}{||\bb{y}_{j}-\bb{y}_{j-1}||}+\frac{\bb{y}_{j+1}-\bb{y}_j}{||\bb{y}_{j+1}-\bb{y}_j||},\bb{n}_{\mathcal{B}}(\bb{y}_j)\Big\rangle=0
\end{equation}
for any $j\in[1{:}i]$.
\end{proposition}
\begin{proof}
From the validity property in (\ref{eq:LawReflection}), we can write
\begin{equation}
\begin{split}
&\Big\langle\frac{\bb{y}_{j}-\bb{y}_{j-1}}{||\bb{y}_{j}-\bb{y}_{j-1}||}+\frac{\bb{y}_{j+1}-\bb{y}_j}{||\bb{y}_{j+1}-\bb{y}_j||},\bb{n}_{\mathcal{B}}(\bb{y}_j)\Big\rangle\\
&=\Big\langle\frac{\bb{y}_{j}-\bb{y}_{j-1}}{||\bb{y}_{j}-\bb{y}_{j-1}||}+\frac{{\bb y}_{j}-P_{\bb{y}_j}(\bb{y}_{j-1})}{\lVert {\bb y}_{j}-P_{\bb{y}_j}(\bb{y}_{j-1})\rVert},\bb{n}_{\mathcal{B}}(\bb{y}_j)\Big\rangle.
\end{split}
\end{equation}
By using the symmetric projection definition in (\ref{eq:SymmProjection}), it is straightforward to see that
\begin{equation}
\begin{split}
\label{eq:KeyIdentity}
\mathbf{y}_j-P_{\bb{y}_j}(\bb{y}_{j-1})=&(\bb{y}_j-\bb{y}_{j-1})\\
&-2\langle\bb{y}_j-\bb{y}_{j-1},\bb{n}_\mathcal{B}(\bb{y}_j)\rangle\bb{n}_\mathcal{B}(\bb{y}_j).
\end{split}
\end{equation}
The latter implies that $||\mathbf{y}_j-\mathbf{y}_{j-1}||=||\mathbf{y}_j-P_{\bb{y}_j}(\bb{y}_{j-1})||$ and 
\begin{equation}
\Big\langle\bb{y}_{j}-\bb{y}_{j-1}+{\bb y}_{j}-P_{\bb{y}_j}(\bb{y}_{j-1}),\bb{n}_{\mathcal{B}}(\bb{y}_j)\Big\rangle=0.
\end{equation}
Thus, the conclusion follows.
\end{proof}

\begin{proposition}
\label{prop:ValidityLines}
Let $(i_1,\ldots,i_k)$ be a sequence of walls indices corresponding to a planar boundary $\mathcal{B}$ consisting of walls $\{W_{i_j}\}$ and $k\in\mathbb{N}$. If $W_{i_j} \cap \mathcal L_{j-1}\cap \mathcal L_j = \{ \bb y_j\}$ for some $\mathbf{y}_j\in \mathcal B$ where $j=1,\dots,k$, then $(\bb s, \bb y_1,\dots,\bb y_k,\bb r)$ is a valid reflection path.
\begin{proof}
Since $\lbrace \bb y_j, \bb y_{j+1}, P_{\bb y_j}(\bb y_{j-1})\rbrace\subset \mathcal L_{j}$, it follows directly that
\begin{align*}
\frac{\bb y_{j+1}-\bb y_j}{\lVert \bb y_{j+1}-\bb y_j \rVert}=\frac{\bb y_{j}-P_{\bb y_j}(\bb y_{j-1})}{\lVert \bb y_{j}-P_{\bb y_j}(\bb y_{j-1}) \rVert}.
\end{align*}
Thus, $(\bb s,\bb y_1,\dots,\bb y_k,\bb r)$ is a valid reflection path.
\end{proof}
\end{proposition}

\section{Hausdorff Measures}
\label{app:HMeasures}
% {\color{olive}In the following lemma, we use a particular partition of a curve in order to motivate the choice of the measure in (\ref{eq:measure}). It is straightforward to see though that it is valid for any partition.
% \begin{lemma}
% \label{lemma:DiracMeasures}
% Let $g:\mathbb{R}^N\to\mathbb{C}$ be a bounded function continuous $\lambda_{\mathbb{R}^N}$-almost everywhere. If $\gamma:[0,1]\to C\subset\mathbb{R}^N$ is a smooth curve $(\lambda_{\mathbb{R}^N}$, a.e.), then the measure defined for $\Gamma\in\mathfrak{B}_N$ by
% \begin{equation}
% \mu_M(\Gamma):=\sum_{m=1}^{M-1}\Big|\Big|\gamma\Big(\frac{m+1}{M}\Big)-\gamma\Big(\frac{m}{M}\Big)\Big|\Big|\delta_{\gamma(\frac{m}{M})}(\Gamma)
% \end{equation}
% converges weakly to $\mathbf{H}^1_C$ as $M\to\infty$ where $\mathbf{H}^1_C=\mathbf{H}^1(\Gamma\cap C)$.
% \begin{proof}
% It is straightforward to see that
% \begin{multline}
% \int_{\mathbb{R}^N}gd\mu_M=\sum_{m=1}^{M}g\Big(\gamma\Big(\frac{m}{M}\Big)\Big)\cdot\Big|\Big|\gamma\Big(\frac{m+1}{M}\Big)-\gamma\Big(\frac{m}{M}\Big)\Big|\Big|\\
% \to\int_{[0,1]}g(\gamma(x))\gamma'(x)\mathrm{d}x=\int_{C}g\mathrm{d}\mathbf{H}^1=\int_{\mathbb{R}^N}g\mathrm{d}\mathbf{H}^1_C
% \end{multline}
% as $M\to\infty$ where the Riemann sum converges to a Riemann integral by Lebesgue's criterion for Riemann integrability \cite[Theorem 7.48]{Apostol:1974aa}. Thus, $\mu_M\Rightarrow\mathbf{H}^1_C$, namely $\mu_M$ converges weakly to $\mathbf{H}^1_C$.}
% \end{proof}
% \end{lemma}
We introduce here some of the background regarding Hausdorff measures. Additionally, in Proposition \ref{prop:DiracMeasures}, we prove a result that illustrates Hausdorff measures as a limiting case of Dirac measures. This is particular insightful when considering denser and denser arrangements of virtual point sources, thus leading to 
virtual sources disposed in a continuous manner through space as considered in this paper.

Let us start by defining the radius of a set $E\subset\mathbb{R}^N$ as
\begin{equation}
\mathrm{rad}(E):=\sup\Big\{\frac{||x-y||}{2}:x,y\in E\Big\}
\end{equation}
with the understanding that $\mathrm{rad}(\emptyset):=0$. For $\mathcal{C}\subseteq\mathcal{P}(\mathbb{R}^N)$, where $\mathcal{P}(\mathbb{R}^N)$ denotes the power set of $\mathbb{R}^N$, we can then set
\begin{equation}
||\mathcal{C}||:=\sup\Big\{2\mathrm{rad}(C): C\in \mathcal{C}\Big\}.
\end{equation}
This last equation is also referred to as the diameter of $\mathcal{C}$. Denoting by $\Omega_N$ the volume of the unit ball in $\mathbb{R}^N$, set $\Omega_s=(1-s)\Omega_N+s\Omega_{N+1}$ for $s\in[N,N+1]$ and $\Omega_0=1$.

Finally, for $\delta>0$ and $s\in[0,\infty)$, the Hausdorff measure is defined as the following limit \cite[Section 8.3.3]{Stroock:2011aa}
\begin{equation}
\mathbf{H}^s(\Gamma)=\lim_{\delta\searrow0}\mathbf{H}^s_\delta(\Gamma)
\end{equation}
where
\begin{multline}
\mathbf{H}^s_\delta(\Gamma):=\inf\Big\{\sum_{C\in\mathcal{C}}\Omega_s\mathrm{rad}(C)^s:\ \mathcal{C}\textrm{ a countable}\\
\textrm{cover of }\Gamma\textrm{ with }||\mathcal{C}||<\delta\Big\}
\end{multline}
It is well known that the restriction of $\mathbf{H}^s$ to $\mathfrak{B}_N$ is a Borel measure \cite[Theorem 8.3.10]{Stroock:2011aa}. Note that $\mathbf{H}^0$ corresponds to the counting measure. The Hausdorff dimension is then defined as
\begin{equation}
\label{eq:HausDim}
\mathrm{Hdim}(\Gamma):=\inf\{s\geq0:\mathbf{H}^s(\Gamma)=0\}
\end{equation}
for any $\Gamma\in\mathfrak{B}_N$.

We prove in the next proposition that a linear combination of appropriately weighted Dirac measures defined for points that get closer and closer on parametrized subsets of the Euclidean space converges weakly to the Hausdorff measure of those parametrized subsets (for example, see \cite{Bogachev:2018aa} for a definition of weak convergence). In particular, we will be defining Dirac measures for a set of points that result from the intersection of a lattice and an open set. This particular choice is convenient to construct Riemann sums on, for example, lines, curves, surfaces, or volumes. Note that different choices are also possible.

Let us denote a lattice in $\mathbb{R}^n$ and parametrized by $\Delta\geq0$ as follows
\begin{equation}
\Lambda_\Delta^n:=\{\mathbf{x}\in\mathbb{R}^n:\mathbf{x}=\sum_{i=1}^na_i\mathbf{e}_i,\ a_i\Delta\in\mathbb{Z}\}.
\end{equation}
where $\{\mathbf{e}_1,\ldots,\mathbf{e}_n\}$ is the canonical basis in $\mathbb{R}^n$. We also need to introduce the Jacobian, denoted by $J\Phi$ for a function $\Phi$. Additionally, for integers $n,p\geq1$ and a linear transformation ${T:\mathbb{R}^p\to\mathbb{R}^n}$, set $\mathcal{J}(T)=\sqrt{\mathrm{det}(T^*T)}$ where $T^*$ is the adjoint of $T$.

Let $C^k(U;\mathbb{R}^N)$ be the space of $k$-times continuously differentiable functions from $U$ into $\mathbb{R}^N$ and $L^1(G;\mu)$ the Lebesgue space of functions $f$ from $G$ into $\mathbb{R}$ for which $|f|$ is $\mu$-integrable. Now, let $V_p(\epsilon)$ be the $p$-dimensional volume of an Euclidean ball of radius $\epsilon>0$.

\begin{proposition}
\label{prop:DiracMeasures}
Consider the open set $U\subseteq\mathbb{R}^p$ for $p\geq1$, and assume further that $\Phi\in C^1(U;\mathbb{R}^N)$ is an injective map such that $\mathrm{rank}(J\Phi(\mathbf{x}))=p$ for every $\mathbf{x}\in U$. Then, the measure defined by
\begin{equation}
\mu_M(\Gamma):=\sum_{\mathbf{x}\in U\cap\Lambda_{M}^N}\delta_{\Phi(\mathbf{x})}(\Gamma)V_p(\epsilon_{\mathbf{x}}),\ \Gamma\in\mathfrak{B}_N
\end{equation}
where 
\begin{equation}
\epsilon_\mathbf{x}:=\min_{\mathbf{x}^\prime\in\Phi(\Lambda_{M}^N)\setminus \mathbf{x}}||\Phi(\mathbf{x}^\prime)-\mathbf{x}||
\end{equation}
converges weakly to $\mathbf{H}^p_{\Phi(U)}$ as $M\to\infty$ for every bounded function ${g\in C^1(\mathbb{R}^N;\mathbb{R})\cap L^1(\Phi(U),\lambda_{\mathbb{R}^N})}$ where $\mathbf{H}^p_{\Phi(U)}={\mathbf{H}^p(\Gamma\cap \Phi(U))}$.
\begin{proof}
It is straightforward to see that
\begin{equation}
\begin{split}
\int_{\mathbb{R}^N}gd\mu_M&=\sum_{b\in U\cap\Lambda_{M}^N}g(\Phi(b))V_p(\epsilon_b)\nonumber\\
&\to\int_{U}g(\Phi(x))\mathcal{J}(J\Phi(x))\mathrm{d}\lambda_{\mathbb{R}^p}\nonumber\\
&=\int_{\Phi(U)}g\mathrm{d}\mathbf{H}^p=\int_{\mathbb{R}^N}g\mathrm{d}\mathbf{H}^p_{\Phi(U)}
\end{split}
\end{equation}
as $M\to\infty$ where the Riemann sum converges to a Riemann integral by Lebesgue's criterion for Riemann integrability \cite[Theorem 7.48]{Apostol:1974aa} and the equality follows from \cite[Theorem 11.25]{Folland:1999aa}. Thus, $\mu_M\Rightarrow\mathbf{H}^p_{\Phi(U)}$, namely $\mu_M$ converges weakly to $\mathbf{H}^p_{\Phi(U)}$.
\end{proof}
\end{proposition}

\begin{example}
\label{ex:interval}
A simple example to illustrate the previous proposition is to consider, as a parametrized curve, the segment $(0,1)$. In particular, this curve is parametrized by $\Phi:(a,b)\to(0,1)\subset\mathbb{R}^N$ for $a,b\in\mathbb{R}$ given by $\Phi(\theta)=(\theta-a)/(b-a)$ assuming $b\neq0$. It is clear that $\Phi$ is an injective transformation in $C^1((a,b);\mathbb{R}^N)$. According to Proposition \ref{prop:DiracMeasures}, if we consider the measure $\mu_{M}(\Gamma):=\frac{1}{M}\sum_{m=1}^{M-1}\delta_{\frac{m}{M}}(\Gamma)$
for $M\geq1$ and $\Gamma\in\mathfrak{B}_N$, we then have that 
\begin{equation}
\int_{\mathbb{R}^N}gd\mu_M\to\int_{\mathbb{R}^N}g\mathrm{d}\mathbf{H}^1_{(0,1)}
\end{equation}
as $M\to\infty$ for every bounded function $g\in C^1(\mathbb{R}^N;\mathbb{R})\cap L^1((0,1),\lambda_{\mathbb{R}^N})$.
\end{example}

\ifCLASSOPTIONcaptionsoff
  \newpage
\fi
\bibliographystyle{IEEEtran}
\bibliography{./ms.bbl}

\end{document}